\documentclass{amsart}

\usepackage[english]{babel}
\usepackage{amsmath}
\usepackage{amsthm}
\usepackage{amssymb}
\usepackage{mathtools}
\usepackage{ragged2e}
\usepackage{etoolbox}
\usepackage{booktabs}
\usepackage{graphicx}

\usepackage{chessfss}

\usepackage{color}
\usepackage{hyperref}
\usepackage{cases}
\usepackage{caption}
\usepackage{enumerate}
\usepackage{enumitem}
\usepackage{multirow}
\hypersetup{
  colorlinks   = true,  
  urlcolor     = blue,  
  linkcolor    = blue,  
  citecolor    = red    
}
\usepackage[font=small,labelfont=bf,justification=justified,format=plain]{caption}
\usepackage[numbers]{natbib}
\usepackage{pgfplots}
\pgfplotsset{compat=1.18}
\usepackage{subcaption}
\usepackage{url}
\newcommand{\ket}[1]{ | #1 \rangle }

\DeclarePairedDelimiter\abs{\lvert}{\rvert}%

\theoremstyle{plain}
\newtheorem{theorem}{Theorem}[section]
\theoremstyle{definition}
\newtheorem{definition}[theorem]{Definition} 

\newtheorem{example}[theorem]{Example}
\newtheorem{lemma}[theorem]{Lemma}
\newtheorem{corollary}[theorem]{Corollary}
\newtheorem{proposition}[theorem]{Proposition}

\usepackage{scalerel,stackengine}
\stackMath
\newcommand\reallywidehat[1]{%
\savestack{\tmpbox}{\stretchto{%
  \scaleto{%
    \scalerel*[\widthof{\ensuremath{#1}}]{\kern-.6pt\bigwedge\kern-.6pt}%
    {\rule[-\textheight/2]{1ex}{\textheight}}
  }{\textheight}%
}{0.5ex}}%
\stackon[1pt]{#1}{\tmpbox}%
}
\usepackage{comment}
\usepackage{float}

\begin{document}

\title{Existence and Characterisation of Bivariate Bicycle Codes}

\author{Jasper J. Postema}
\email{jasperjpostema@gmail.com}

\author{Servaas J.J.M.F. Kokkelmans}

\date{\today}

\begin{abstract}
Encoding quantum information in a quantum error correction (QEC) code offers protection against decoherence and enhances the fidelity of qubits and gate operations. One of the fundamental challenges of QEC is to construct codes with \textit{asymptotically good} parameters, i.e. a non-vanishing rate and relative minimum distance. Such codes provide compact quantum memory with low overhead and enhanced error correcting capabilities, compared to state-of-the-art topological error correction codes such as the surface or colour codes. Recently, bivariate bicycle (BB) codes have emerged as a promising candidate for such compact memory, though the exact tradeoff of the code parameters $[[n,k,d]]$ remained unknown. In this Article, we explore these codes by leveraging their ring structure, and predict their dimension as well as conditions on their existence. Finally, we highlight \textit{asymptotic badness}. Though this excludes this subclass of codes from the search towards practical good low-density parity check (LDPC) codes, it does not affect the utility of the moderately long codes that are known, which can already be used to experimentally demonstrate better QEC beyond the surface code.
\end{abstract}

\maketitle

\section*{Introduction}

Quantum error correction will be a crucial ingredient for large-scale \textit{fault-tolerant} quantum computing in the (near) future \cite{shor, danielgottesman, toric, terhal}. Encoding quantum information in error-correcting codes enables logical error rates to be reduced to arbitrarily low levels, provided the code distance $d$ is sufficiently large, with the first experimental realisations having been demonstrated recently \cite{googleAI2}. Stabiliser codes that outperform the surface code generally require a lot of overhead. Codes with the best performance are known as \textit{asymptotically good}, as the number of logical qubits $k$ and distance $d$ scale asymptotically linearly in the number of physical qubits $n$, as opposed to \textit{asymptotically bad} codes where either $k/n$ or $d/n$ vanishes asymptotically. In particular, low-density parity check (LDPC) codes can achieve better parameters by trading off qubit overhead for long-range connectivity requirements. Despite the existence of asymptotically good quantum LDPC codes \cite{ldpc}, the search for good LDPC codes that can be practically implemented in near-term hardware remains a severe challenge.\\

The Calderbank-Shor-Steane (CSS) construction provides an explicit construction to produce\textit{ quantum codes }\cite{CS,S}. This construct takes a pair of classical codes $(C_1, C_2)$ over a finite field as its input and produces a quantum code of $n$ physical qubits that can correct phase and amplitude errors independently. Random binary CSS constructions will not yield good codes because the orthogonality condition
\begin{equation}\label{eq:orth}
    H_XH_Z^\top = 0
\end{equation}
on the parity check matrices imposes a strict restriction on the existence of non-trivial random constructions \cite{alberto}. Therefore, many codes are constructed from a top-down approach and presume a generalised bicycle ansatz \cite{bicycle} that satisfies condition Eq. (\ref{eq:orth}).\\

Recently, \textit{bivariate bicycle} (BB) codes have been proposed as a promising candidate for LDPC quantum codes with possibly a practical transpilation scheme on near-term quantum computing platforms \cite{bravyi, terhalBB, chenplanar}. Such codes belong to the family of\textit{ lifted product codes,} which contains the first asymptotically good quantum codes \cite{ldpc, LP}. It's unknown what the exact trade-off is between the code parameters $[[n,k,d]]$. For any BB code of 'size' $\ell\times m$, one can use blind brute force to find non-trivial codes, though this 
is computationally expensive. Instead, we consider these codes as ideals in a \textit{polynomial quotient ring} generated by two cyclic groups of size $\ell$ and $m$. Studying two cases, where $\ell$ and $m$ are either odd and mutually coprime, yielding \textit{coprime BB codes} \cite{coprime, linpryadko}, or where they are free to take any integer value, yielding \textit{quasi-cyclic BB codes}, we study their code parameters $[[n,k,d]]$, examine under what circumstances these codes exist, and demonstrate \textit{asymptotic badness} as the number of physical qubits $n$ grows to infinity for any sequence of BB codes, showing that the abelian nature of BB codes forms the limiting factor towards asymptotic goodness.\\

In this Article, we characterise the fundamental properties of BB code parameters. In Sec. \ref{sec:coding}, we lay out the relevant background theory of classical and quantum codes. We introduce BB codes in Sec. \ref{sec:bb}, calculate their code dimension and upper bounds for their distance, and we provide conditions on the existence of non-trivial BB codes in Sec. \ref{sec:construct}. We discuss BB code performance in Sec. \ref{sec:performance} by comparing to the rotated surface code, and we highlight the smallest possible BB codes. A summary of our work is presented in Sec. \ref{sec:summary}.

\section{Preliminaries - Coding Theory}\label{sec:coding}

In this Section, we lay out all the relevant mathematical theory to understand our findings. In particular, the structure of polynomial quotient rings plays a crucial role in understand the parameter tradeoff for BB codes. Throughout this Article, we define $[n] = \{1,\ldots, n\}$. Moreover, $\mathbb{F}_q$ is the finite Galois field of $q$ elements, where $q=p^s$ is a prime power, and $\mathbb{F}_q^\times$ denotes all the non-zero elements of $\mathbb{F}_q$. The greatest common divisor is denoted as gcd, the least common multiple as lcm, $\mathbb{Z}_n:=\mathbb{Z}/n\mathbb{Z}$ is shorthand for the ring of integers modulo $n$, and $x\mid y$ means that $x$ is a divisor of $y.$ 

A linear code $C\subseteq\mathbb{F}_q^n$ is a non-empty linear subspace of the $n$-dimensional vector space $\mathbb{F}_q^n$. The vectors $c\in C$ are called codewords. The code dimension $k$ is equal to its dimension as a linear subspace over $\mathbb{F}_q$, and is related to the code cardinality through $|C|=q^k$. Linear codes can be described by a generator matrix $G\in\mathbb{F}_q^{k \times n}$ that generates all codewords through $uG$ for all $u\in\mathbb{F}_q^k$, and a parity check matrix $H\in\mathbb{F}_q^{(n-k)\times n}$ of full rank $n-k$ that annihilates all codewords: $GH^\top = 0.$ The Hamming distance between two vectors $x,y\in\mathbb{F}_q^n$ is given by
\begin{equation*}
    d(x,y) = \abs{\{1\leq i\leq n \mid x_i\neq y_i\}}
\end{equation*}
and corresponds to a distance metric over $\mathbb{F}_q^n$. The Hamming weight $\omega_H(c)$ is the number of non-zero entries of $c\in C$. Then, the minimum distance $d$ of $C$ is the minimum Hamming distance between distinct codewords. For linear codes,
\begin{equation*}
    d = \text{min} \{\omega_H(c) \; | \; c \in C \setminus \{\textbf{0}\}\}
\end{equation*}
is the minimum Hamming weight of any non-zero codeword in the code $C$. Such a code can correct for $d-1$ erasures and $\lfloor\frac{d-1}{2}\rfloor$ errors. Binary codes are ubiquitous, though codes over binary extension fields such as Reed-Solomon codes or codes defined over the Riemann-Roch space of an algebraic curve function field have played a prominent role in classical error correction \cite{reedsolomon, algebraicgeometry}.\\

The CSS construction is a systematic way to produce quantum codes from a pair of classical codes, derived independently by Calderbank, Shor and Steane \cite{CS,S}. It yields families of codes that can independently correct for amplitude and phase errors.

\begin{definition}[\cite{CS, S} CSS code]
    Let $C_2\subseteq C_1 \subseteq \mathbb{F}_q^n$ be two linear codes. Let $\zeta = \exp\left(2\pi i/p\right)$ be a primitive complex $p$-th root of unity and let $\text{tr}_{\mathbb{F}_q/\mathbb{F}_p}:=\text{tr}$ denote the field trace map from $\mathbb{F}_q$ to its prime subfield $\mathbb{F}_p$. For any vector $w\in\mathbb{F}_q^n$ we define the qudit state
    \begin{equation*}
        \ket{c_w} := \frac{1}{\sqrt{|C_1|}}\sum_{c\in C_1}\zeta^{\text{tr} ( c,w )}\ket{c},
    \end{equation*}
    where $(\cdot, \cdot)$ denotes the (Euclidean) inner product over $\mathbb{F}_q$. Then, Calderbank's and Shor's definition of a quantum code associated with the pair $(C_1, C_2)$ is
    \begin{equation*}
        Q^\text{CS}(C_1, C_2) = \{\ket{c_w}\mid w\in C_2^\perp\},
    \end{equation*}
    while Steane's variant is given by
    \begin{equation*}
        Q^\text{S}(C_1, C_2) = \{\ket{w+C_2} \mid w\in C_1\},
    \end{equation*}
    where we define $\ket{w+C_2}:=\frac{1}{\sqrt{|C_2|}}\sum_{c\in C_2} \ket{w+c}$. There exists a bijection between codewords in both definitions, which can be used to show that $Q^\text{CS}(C_1, C_2) = Q^\text{S}(C_2^\perp, C_1^\perp)$, thus both are equivalent definitions.
\end{definition}
Such a code is said to be an $[[n, k, d]]_q$-code, using $n$ qudits to encode $k$ \textit{logical qudits} with distance $d$. From now on, we restrict ourselves to the case of qubits: $q = 2$, and drop the index on $[[n,k,d]]$. Naturally, quantum error correction codes $\mathcal{Q}$ are images of the injective and norm-preserving mapping $\Phi: \mathfrak{H}^k\xhookrightarrow{}\mathfrak{H}^n$, where $\mathfrak{H}^z$ denotes a $z$-qubit Hilbert space of dimension $2^z$. The dimension of the CSS code is calculated as
\begin{equation*}
    k = n - \text{rank}_{\mathbb{F}_2} H_X - \text{rank}_{\mathbb{F}_2} H_Z,
\end{equation*}
while the minimum distance is given by
\begin{equation*}\label{eq:distance}
    d = \text{min}_{c \in C_Z^\perp / C_X \;\text{or}\; c \in C_X^\perp / C_Z} \,\omega_H(c) = \text{min}(d_X, d_Z).
\end{equation*}
As in the classical sense, such a code can correct for $d-1$ erasures and $\lfloor\frac{d-1}{2}\rfloor$ errors, though the gauge symmetry of the stabilisers allows for a limited number of configurations of errors exceeding the code distance to still be corrected. If $k=0$ (no logical qubits) or $d\leq2$ (no error correcting capabilities), a code is said to be \textit{trivial}. An important and experimentally favourable property of quantum codes is to be \textit{low-density parity check} (LDPC). Let $\{C^\text{CSS}_i\}_{i\in\mathbb{N}}$ be a sequence of CSS codes, then it is LDPC if and only if the row and column weights of its parity check matrices are bounded by $i$-independent constants $\Delta_\text{col}, \Delta_\text{row}.$ Popular choices include $\Delta = 4$ (e.g. surface, toric \cite{toric}, Steane code), $\Delta = 6$ (e.g. honeycomb code, BB codes in this article) and $\Delta = 8$ (e.g. tetrahedral code \cite{tetrahedral, tetrahedral2}).\\

Every logical qubit of a quantum code is endowed with a set of \textit{logical operators} with Hamming weight equal to  or exceeding $d$. The logical Pauli operators $X_L$ and $Z_L$ can be identified with the (co)homology groups
\begin{equation*}
    \text{Hom}_X := \text{ker}\, H_Z\setminus \text{im}\,H_X^\top
\end{equation*}
and
\begin{equation*}
    \text{Hom}_Z := \text{ker}\,H_X\setminus\text{im}\,H_Z^\top.
\end{equation*}

In this paper, we study codes that are closely related to (quasi-)cyclic codes. Here, we briefly review their definitions and structure:

\begin{definition}[(Quasi-)cyclic codes]
    Let $\text{gcd}(q,n)=1$. A linear code $C\subseteq\mathbb{F}_q^n$ is \textit{cyclic} if for every codeword $c=(c_1,\ldots,c_n)\in C$, we have that the cyclic shift operator $T$ preserves the code structure, i.e. $Tc:=(c_n,c_1,\ldots,c_{n-1})\in C$. This operator is subject to $T^n=1.$ Such codes are generated by a univariate polynomial $g(x)\mid x^n-1$, and are ideals in
    \begin{equation*}
        \frac{\mathbb{F}_2[x]}{\langle x^n-1\rangle}.
    \end{equation*}
    A linear code $C\subseteq\mathbb{F}_q^{mn}$ is a \textit{quasi-cyclic} code of index $n$ if for every codeword $c=(c^{(1)}_1,\ldots,c^{(1)}_n,\ldots,c^{(m)}_1,\ldots,c^{(m)}_n)\in C$, we have that the cyclic shift operator $T$ preserves the code structure, i.e. $Tc=(c^{(1)}_n,c^{(1)}_1,\ldots,c^{(1)}_{n-1},\ldots,c^{(m)}_n,c^{(m)}_1,\ldots,c^{(m)}_{n-1})\in C$. This operator is subject to $T^n=1.$ The ideal $\langle x^n-1\rangle$ fixes $\mathbb{F}_q^{mn}$, thus we can view $C\subseteq\mathbb{F}_q^{mn}$ as an $\mathbb{F}_q[x]/\langle x^n-1\rangle$-submodule of $\mathbb{F}_q^{mn}.$ A special case is when each block of $m$ symbols itself has cyclic invariance, a subcase of \textit{constacyclic} codes, which are ideals of the ring
    \begin{equation*}
        \frac{\mathbb{F}_2[x_1,x_2]}{\langle x_1^m-1,x_2^n-1\rangle}.
    \end{equation*}
\end{definition}

Since $x^n-1$ and its formal derivative have no common roots if $\text{gcd}(q,n)=1$, we can find a unique factorisation into square-free irreducible monic polynomials called \textit{minimal polynomials}. It is easier to understand cyclic codes in terms of these minimal polynomials, compared to the ideal $x^n-1$. Under the Chinese remainder theorem, we can find a ring isomorphism that maps the polynomial quotient ring to smaller constituents:

\begin{theorem}[Chinese remainder theorem]
   Consider a ring $\mathcal{R}$ and fix two-sided ideals $I_1,\ldots, I_N$ that are pairwise coprime, then the Chinese remainder theorem (CRT) states that
        \begin{equation*}
            \frac{\mathcal{R}}{\bigcap_i I_i}\cong \frac{\mathcal{R}}{I_1} \oplus \cdots \oplus \frac{\mathcal{R}}{I_N}.
        \end{equation*}
\end{theorem}

\begin{example}
    Let $\text{gcd}(q,n)=1$, and take $\mathcal{R}$ to be the polynomial quotient ring $\mathbb{F}_q[x]/\langle x^n-1\rangle$. If we expand $x^n-1$ into a set of monic square-free irreducible polynomials $f_1(x)\cdots f_\eta(x)$, then the CRT states that
    \begin{equation*}
        \frac{\mathbb{F}_q[x]}{\langle x^n-1\rangle} \cong \bigoplus_{i \in [\eta]} \frac{\mathbb{F}_q[x]}{\langle f_i(x)\rangle }
    \end{equation*}
    is a ring isomorphism given by $p(x)\mapsto \{p(x) \text{ mod } f_i(x)\}_{i=1}^\eta$.
\end{example}

Minimal polynomials can either come in reciprocal pairs, or they are self-reciprocal, where we define reciprocity: 

\begin{definition}[Reciprocity]
    The \textit{reciprocal} or \textit{conjugate} of a polynomial $p(x)$, denoted as $p^\star(x)$, is defined as
    \begin{equation*}\label{eq:conj}
        p^\star(x) = {x^{\text{deg} (p)}}p\left(\frac{1}{x}\right).
    \end{equation*}
    We call a polynomial \textit{symmetric} or \textit{self-reciprocal} if $p(x)=p^\star(x)$, and asymmetric otherwise. Two polynomials $p(x)\neq q(x)$ such that $p(x)=q^\star(x)$ are called a conjugate pair. Conjugation is distributive, i.e. $(p(x)q(x))^\star = p^\star(x)q^\star(x)$.
\end{definition}

One can show that every minimal polynomial is either self-reciprocal $(g_i(x))$ or part of a self-reciprocal doublet $(h_i(x), h_i^\star(x))$, so that
\begin{equation*}
    x^n-1= g_1(x)\cdots g_s(x)h_1(x)h_1^\star(x)\cdots h_t(x)h_t^\star (x),
\end{equation*}
where $\eta=s+2t$. Let $\mathbb{G}_i:=\mathbb{F}_q[x]/\langle g_i(x)\rangle$, $\mathbb{H}_j':=\mathbb{F}_q[x]/\langle h_j(x)\rangle$ and $\mathbb{H}_j'':=\mathbb{F}_q[x]/\langle h_j^\star(x)\rangle$. By the CRT, the decomposition becomes
\begin{equation*}
    \frac{\mathbb{F}_2[x]}{\langle x^n-1\rangle}\cong \left(\bigoplus_{i=1}^s \mathbb{G}_i\right)\oplus\left(\bigoplus_{j=1}^t \left(\mathbb{H}_j'\oplus \mathbb{H}_j''\right)\right).
\end{equation*}
The number of irreducible minimal polynomials is equal to
\begin{equation*}
    \eta = \sum_{d\mid n}\frac{\varphi(d)}{\text{ord}_d(2)},
\end{equation*}
where $\varphi(\cdot)$ is the Euler totient function, and $\text{ord}_d(2)$ is the multiplicative order of 2 modulo $d$, i.e. the smallest integer $e$ such that $2^e\equiv 1$ mod $d$.

\begin{definition}[Cyclotomic polynomial and coset]
    Let $\mathbb{F}_q$ be the finite field of $q$ elements and let $n$ be a positive integer such that $\text{gcd}(q,n) = 1$. The $n$-th \textit{cyclotomic polynomial} is the unique monic  polynomial that is a divisor of $x^n-1$, but not of $x^m-1$ for all $m<n$. Let $\zeta$ be an $n$-th root of unity in an extension field of $\mathbb{F}_q$. Then we can write
    \begin{equation*}
        \Phi_n(x) = \prod_{\substack{i\leq k \leq n \\ \text{gcd}(k,n)=1}}\left(x-\zeta^k\right).
    \end{equation*}
    Over $\mathbb{F}_q$, these cyclotomic polynomials are not necessary irreducible. We can write cyclotomic polynomials in terms of \textit{minimal polynomials} $f_j(x)$ which are $\mathbb{F}_q$-irreducible. The roots of a minimal polynomial can be written as $\{\zeta^i\}_{i\in C_s}$, where $C_s$ is the \textit{cyclotomic coset}
    \begin{equation*}
        C_s :=\{s,qs,q^2s,\ldots,q^{\text{ord}_n(2)-1}s\}.
    \end{equation*}
\end{definition}

Furthermore, we require the following two definitions:

\begin{definition}[Polynomial order]
    The \textit{order} of a univariate polynomial $p(x)$ is the smallest integer $e$ such that $p(x) \mid 1+x^e$. A cyclotomic polynomial $\Phi_n(x)$ has order $n$ by definition, and so do all of its minimal polynomials.
\end{definition}

\begin{definition}[Splitting field]
    The \textit{splitting field} of the $\mathbb{F}_2$-polynomial $1+x^n$ is the smallest field extension of $\mathbb{F}_2$ over which the polynomial decomposes into linear factors, given by $\mathbb{F}_{2^{\text{ord}_n(2)}}$.
\end{definition}

\section{Bivariate Bicycle codes}\label{sec:bb}

Bivariate bicycle (BB) codes have recently been proposed as a candidate for compact quantum memory \cite{bravyi}. They admit a \textit{representation} in terms of cyclic matrices, which we shall lay out here, though we stress that we do not need this exact representation to understand the structure of these codes nor to find the code parameters $k$ and $d$. We define the $\ell\times \ell$ cyclic shift operator $S_\ell\in\mathbb{F}_2^{\ell\times\ell}$ component-wise as follows:
\begin{equation*}
    \left[S_{\ell}\right]_{ij} = 
\begin{cases}
1 & \text{if } j = i + 1 \text{ mod } \ell, \\
0 & \text{otherwise},
\end{cases}
\end{equation*}
and define the bivariate variables $x$ and $y$ as tensor products of the shift and identity matrices:
\begin{equation*}
    x = S_\ell\otimes \mathbb{I}_m \quad \text{and} \quad y = \mathbb{I}_\ell \otimes S_m. 
\end{equation*}
These variables define bivariate $\mathbb{F}_2$-polynomials\footnote{When we capitalise $A,B$, we refer to the matrix representation, while $a,b$ refer to the polynomials. The kernel of $A$ corresponds to the ring annihilator of $a$, for example.}
\begin{equation*}
    A:=a(x, y) = x^{a_1} + y^{a_2} + y^{a_3}\quad\text{and}\quad
    B:=b(x, y) = y^{b_1} + x^{b_2} + x^{b_3}
\end{equation*}
under the identification $x^\ell=1$ and $y^m=1$. This \textit{trinomial ansatz} is inspired by near-term hardware constraints, and enforces the code to be LDPC. Then, the parity check matrices of these codes are given by means of horizontal stacking, yielding the \textit{generalised bicycle} ansatz
\begin{equation*}\label{eq:ansatz}
    H_X = [A | B] \quad \text{and} \quad H_Z = [B^\top | A^\top],
\end{equation*}
a structure that automatically satisfies $H_X H_Z^\top = 0$: the required commutative relationship between the parity check matrices of a CSS code.  As a consequence, every stabiliser is at most weight-6 regardless of the choice of polynomials\footnote{The stabiliser weight is exactly 6 if none of the terms in $A(x,y)$ and $B(x,y)$ cancel out, which we will later see is required for the code dimension to be $>0$.}. The transposed of a polynomial $p(x,y)$ can be viewed as the mapping
\begin{equation*}
    x^i\xrightarrow{\top} x^{\ell-i}\quad \text{and} \quad y^j\xrightarrow{\top} y^{m-j}
\end{equation*}
that generalises \textit{reciprocity} to bivariate polynomials. Key to what makes BB codes favourable candidates for implementation on physical hardware on near-term quantum computers is their simple structure. Ref. \cite{bravyi} elaborates on the graph structure of the Tanner graph that underlies BB codes, and shows it has a thickness 2, so that we can decompose it as
\begin{equation*}
    E = E_1 \sqcup E_2,
\end{equation*}
where $E_i$ denotes a set of edges spanning a thickness-1 (planar) graph. This motivates us to view BB codes as having a \textit{wheel structure}, as highlighted in Fig. \ref{fig:bb}. Transpilation schemes that leverage this graph decomposition may possibly be efficient.\\

\begin{figure*}
    \centering{
    \includegraphics[width=0.98\textwidth]{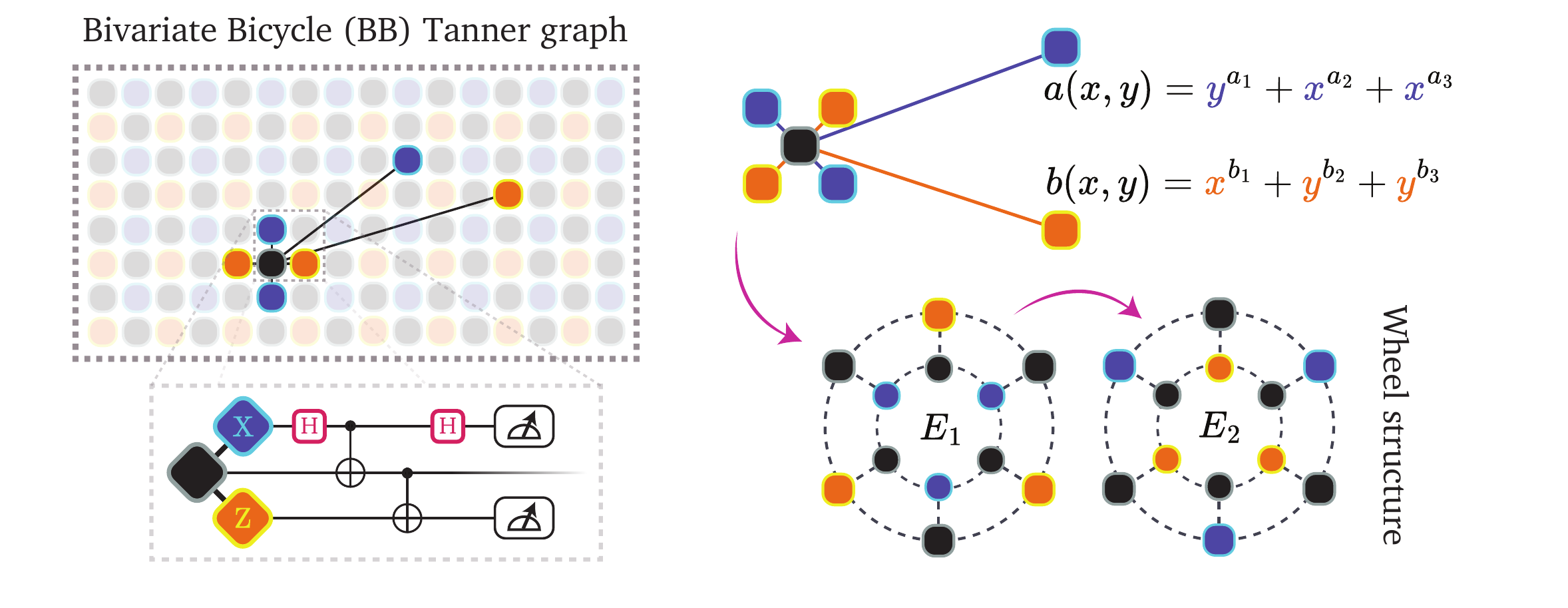}
    \captionsetup{justification=Justified}
 \caption{Tanner graph of a BB code, where \textbf{black} squares denote data qubits, and blue/golden squares denote stabiliser qubits. Two polynomials $a,b$ define the overall connectivity. The trinomial ansatz enforces the code to be LDPC with weight-6 stabilisers. There always exists a graph permutation such that one can rearrange the grid in such a way that every qubit has four nearest neighbours in the 2D plane, and two non-local neighbours across the grid. Though the overall Tanner graph $\mathcal{G}=(E,\mathcal{V})$ is non-planar, there exists a graph/edge partitioning $E = E_1\sqcup E_2$ such that every qubit is part of at most two wheel-like planar structures. Stabilisation is facilitated by the usual 2-qubit interactions (the CNOT gate).}
    \label{fig:bb}
    }
\end{figure*}

The generalised bicycle ansatz structure allows us to extract $k$ and $d$ from the matrices $A$ and $B$.  In particular,
\begin{equation*}
    k = 2\,\text{dim}(\text{ker}\, A \cap \text{ker}\, B)    
\end{equation*}
and the code distance $d$ is
\begin{equation*}
    d = \text{min}\{\omega_H(c) \mid c\in\text{ker}\, H_X\setminus\text{im}\,H_Z^\top \}.
\end{equation*}
Henceforth, we will refer to the octet $\mathfrak{o} = \{\ell, m, a_1, a_2, a_3, b_1, b_2, b_3\}$ as the \textit{constructors} of the BB code. Since there is plenty of freedom in the choice of constructors, the natural question arises:\\

\begin{center}
\textit{What constructors lead to non-trivial BB codes?}
\end{center}
\vspace{0.2cm}

This question will be answered throughout the remainder of this paper. One thing to note is that different codes might give rise to the same parameters, i.e. the notion of \textit{equivalent codes} that have a different constructor octet $\mathfrak{o}$, but equivalent code parameters \cite{linpryadko}. If two polynomials are \textit{equivalent}, then the codes they generate are also equivalent:

\begin{definition}[Equivalence]\label{def:equiv}
    Two polynomials $p(x,y)$ and $q(x,y)$ are said to be \textit{equivalent} ($p\equiv q$) if they are equal up to multiplication by a factor of $x$ and/or $y$. For example, in the bivariate polynomial ring quotiented by $\langle x^6-1,y^6-1\rangle$, 
    \begin{equation*}
        xy^3+x^3\equiv 1+x^4y^3.
    \end{equation*}
    Another example in the univariate ring modulo the ideal $\langle z^{15}-1\rangle$: \begin{equation*}
        1+z+z^2\equiv 1+z+z^{14},
    \end{equation*}
    as they differ by a factor $z^{14}=z^{-1}$. The code parameters of BB codes are invariant under equivalence of the constructor trinomials.
\end{definition}

BB codes are a subclass of \textit{two-block group algebra} (2BGA) codes \cite{linpryadko}, defined over the (quasi-)cyclic commutative group algebra $\mathcal{G} = \mathbb{Z}_\ell\times\mathbb{Z}_m$. Their parity check matrices can be interpreted as \textit{quasi-cyclic codes} or \textit{constacyclic codes}, which can be understood as ideals in the polynomial quotient ring
\begin{equation*}
    \mathcal{S} = \mathbb{F}_2[\mathcal{G}] = \frac{\mathbb{F}_2[x,y]}{\langle x^\ell-1, y^m-1\rangle},
\end{equation*}
the ring of bivariate polynomials over $\mathbb{F}_2$ in $x$ and $y$, under the identification $x^\ell=y^m=1$, equipped with the standard basis of monomials $\{x^\mu y^\nu\}$ for $\mu\in [\ell],\nu\in[m]$. In contrast to the definitions of $k$ and $d$ in Ref. \cite{bravyi}, which employ linear algebra in $\mathbb{F}_2$, we can ask the question if these parameters can be understood from more elegant principles, in particular the ring structure of $\mathcal{S}$. If $\ell$ and $m$ are coprime, yielding \textit{coprime BB codes}, then the group algebra is semi-simple by virtue of the identity
\begin{equation*}
    \mathbb{Z}_\ell\times\mathbb{Z}_m \cong \mathbb{Z}_{\ell m} \quad\text{if }\text{ gcd}(\ell,m)=1, 
\end{equation*}
and can therefore be understood as a code generated by a univariate polynomial. If $\text{gcd}(\ell,m)\neq 1$, we have \textit{quasi-cyclic BB codes}\footnote{This is a bit of a misnomer since the codes that generate the CSS quantum code are themselves not quasi-cyclic. Nonetheless we stick to this terminology.}.

In the remainder of this Section, we will prove 3 properties. First, coprime BB codes can always be understood as an ideal of a univariate polynomial quotient ring. Their dimension can easily be determined from the greatest common divisor of the parity check generators. Both of these properties have already been proven to a degree in Ref. \cite{coprime}, so here we provide a more precise characterisation. Then we characterise the code dimension of one-sided and two-sided repeated root codes, which are sub-categories of quasi-cyclic BB codes that we will define later. We employ upper bounds on the code distance due to recent results by Bravyi and Terhal \cite{bravyidim,bravyidim2}, Wang and Pryadko \cite{generalisedbicycle}, and Arnault et al. \cite{generalisedbravyiterhal}, and demonstrate asymptotic badness. These theorems comprise a full characterisation of BB codes, and are then used to find new codes. After this Section, we characterise under which circumstances the constructors $\mathfrak{o}$
yield non-trivial codes.

\begin{theorem}[Univariate generator] Let $\ell$ and $m$ be coprime and odd (i.e. coprime with the field characteristic). Then, BB codes allow for a univariate representation. More concretely, there exists a bijection $\psi$ such that
\begin{equation*}
    \psi: \frac{\mathbb{F}_2[x,y]}{\langle x^\ell-1, y^m-1\rangle} \to \frac{\mathbb{F}_2[z]}{\langle z^{\ell m}-1\rangle},
\end{equation*}
where $z=\psi(xy)$. This mapping is carried out by
\begin{equation*}
    \psi: x \mapsto z^{m^{-1_\ell}m},\quad  y \mapsto z^{\ell^{-1_m}\ell}
\end{equation*}
where $m^{-1_\ell}$ is the multiplicative inverse in $\mathbb{Z}_\ell$ (and $\ell^{{-1}_m}$ is defined in a similar fashion).
\end{theorem}

\begin{proof}
    Since $\langle x \rangle$ and $\langle y \rangle$ are cyclic groups of order $\ell$ and $m$ respectively, one can see that $\langle xy \rangle$ is a cyclic group of order $\ell m$. Let $x=z^t$. Since there is no dependency on $y$, $t$ must be a multiple of $m$, so we can write it as $\lambda m$ for some $\lambda \in \mathbb{N}$. Then, $\lambda$ must satisfy $\lambda \ell = 1$ mod $m$, thus $\lambda = \ell^{-1}$ mod $m$. Now we prove this inverse exists and is well-defined. If $\ell$ is a prime (power), then $\mathbb{Z}_\ell$ is a finite field (extension), in which a multiplicative inverse exists by the definition of a field. If $\ell$ is not a prime-power, it is the ring of integers $\{0,1,\cdots,\ell-2,\ell-1\}$. By the Chinese remainder theorem, we can always decompose this as follows: let $\ell = p_1^{s_1}\ldots p_N^{s_N}$ be the prime decomposition of $\ell$. Then,
\begin{equation*}
    \frac{\mathbb{Z}}{\ell\mathbb{Z}} \cong \frac{\mathbb{Z}}{p_1^{s_1}\mathbb{Z}}\times \cdots \times \frac{\mathbb{Z}}{p_N^{s_N}\mathbb{Z}}.
\end{equation*}
Since $m$ and $\ell$ are coprime, they share no prime factors. Thus $m$ is not a zero of any of the constituent fields, and a unique  multiplicative inverse exists. In a very similar fashion, we can define how $\psi$ maps $y$ to $\psi(y)\in\frac{\mathbb{F}_2[z]}{\langle z^{\ell m}-1\rangle}$.
\end{proof}

The mapping from the trinomial ansatz (\ref{eq:ansatz}) to the univariate representation is \textit{injective}, so that any arbitrary univariate trinomial is not necessarily restricted to the original ansatz. For example, we could pick an arbitrary $\mathbb{F}_2[z]$-polynomial, which maps to $1+xy+x^2$ under $\psi^{-1}$. \textit{However}, throughout the rest of the Article, we will continue to stick to arbitrary univariate trinomials as \textit{the ansatz} for our codes, regardless whether they satisfy the original bivariate ansatz (\ref{eq:ansatz}) or not. We will also equip the ring $\mathcal{R}=\frac{\mathbb{F}_2[z]}{\langle z^{\ell m}-1\rangle}$ with the standard basis of monomials $\{z^\mu\}$ for $\mu\in[\ell m].$ Because multiplication of the trinomials by a factor $z$ does not change the fundamental structure of $a(z)$ or $b(z)$, we will henceforth set their lowest power to 0 without loss of generality. Given these polynomials, the code dimension readily follows:

\begin{theorem}[\cite{bicycle, coprime,degeneratequantum}, Coprime code dimension]
    The code dimension of a coprime BB code is given by $k^\text{BB} = 2\,\textnormal{deg} \, g(z)$, where $g(z)$ is the generator of the parity check matrices
    \begin{equation*}
        g(z) = \textnormal{gcd}\left(a(z), b(z), z^{\ell m}-1\right).
    \end{equation*}
\end{theorem} 

\begin{proof} Let $u(z), v(z) \in \mathcal{R}$ be two polynomials. Then, any element in $\text{im}\, H_X$ can be written as the block vector
\begin{equation*}
    [a(z)u(z)\mid b(z)v(z)],
\end{equation*}
where $\mid$ denotes concatenation. Since $\frac{\mathbb{F}_2[z]}{\langle z^{\ell m}-1\rangle}$ is a univariate polynomial quotient ring, the image of $H_X$ forms a principal ideal, and is therefore generated by $g(z)$. For cyclic codes, we find $\text{rank}_{\mathbb{F}_2}H_X = \ell m - \text{deg}\,g(z).$ As $k^\text{BB} = n - 2\cdot\text{rank}_{\mathbb{F}_2}H_X$, the coprime BB code dimension yields
\begin{equation*}
    k^\text{BB} = 2\,\text{deg} \, g(z).
\end{equation*}
This concludes the proof.
\end{proof}

For the quasi-cyclic case, we must find a suitable basis for the ideals generated by $a(x,y)$ and $b(x,y)$ first. First, we consider \textit{one-sided repeated root} codes, where we set $\ell$ to be an odd integer ($\text{gcd}(\ell, 2) = 1$), and $m$ is any positive integer. Then we fix an isomorphism from $\mathcal{S}$ to its constituents (under the Chinese remainder theorem) as follows \cite{LP}: 

\begin{lemma}[$\mathcal{S}$-ring isomorphism]
    Let $x^n-1=\prod_{j=1}^\eta f_j(x)$, where for each $1\leq j\leq \eta$, $f_j(x)$ is a monic irreducible polynomial in $\mathbb{F}_2[x]$. Define $\mathcal{S}_j:=\frac{\mathbb{F}_2[x,y]}{\langle f_j(x), y^m-1\rangle}$. Then,
    \begin{equation*}
        \psi:=(\psi_1,\ldots,\psi_\eta):\mathcal{S}\to\bigoplus_{j=1}^\eta \mathcal{S}_j
    \end{equation*}
    is a valid ring isomorphism, where we define $\psi_j:\mathcal{S}\to\mathcal{S}_j$ to act on polynomials $p(x,y)$ as
    \begin{equation*}
        \psi_j(p(x,y))=p(x,y)\text{ mod } f_j(x).
    \end{equation*}
\end{lemma}
\begin{proof}
    The image of $H_X$ forms an ideal in the ring $\mathcal{S}$. Since $\mathcal{S}\to\bigoplus_{j=1}^\eta \mathcal{S}_j$, the image of $H_X$ is a direct sum of ideals in $\mathcal{S}_j$ according to
    \begin{equation*}
        \mathcal{S}=\frac{\mathbb{F}_2[x,y]}{\langle x^\ell-1,y^m-1\rangle}\cong\bigoplus_{j=1}^\eta\frac{\frac{\mathbb{F}_2[x]}{\langle f_j(x)\rangle}[y]}{\langle y^m-1\rangle}:=\bigoplus_{j=1}^\eta \mathcal{S}_j.
    \end{equation*}
    Defining the mapping $\psi_j:\mathcal{S}\to\mathcal{S}_j$, we see that $\psi=(\psi_1,\ldots,\psi_\eta)$ forms a valid ring isomorphism.
\end{proof}

\begin{lemma}[Parity check generator]
    Codes, defined over the ring $\mathcal{S}$, with a parity check matrix $H=[A\mid B]$ have a parity check generator
    \begin{equation*}
        H=\langle g_1(x,y)\prod_{i\neq 1} f_i(x),\ldots, g_\eta(x,y)\prod_{i\neq \eta} f_i(x)\rangle,
    \end{equation*}
    where
    \begin{equation*}
        g_i(x,y)=\text{gcd}(\psi_i(a(x,y)),\psi_i(b(x,y)),y^m-1)\quad\text{in }\frac{\mathbb{F}_2[x]}{\langle f_i(x) \rangle}[y]. 
    \end{equation*} 
\end{lemma}
\begin{proof}
    See Ref. \cite{constacyclic}.
\end{proof}

\begin{theorem}[Quasi-cyclic code dimension]
    The code dimension of a quasi-cyclic BB code is given by $k^\text{BB} = 2\sum_{j=1}^\eta  \textnormal{deg}_y g_j(x,y)\textnormal{deg}_x f_j(x)$, where we defined the generator
    \begin{equation*}
        H=\langle g_1(x,y)\prod_{i\neq 1} f_i(x),\ldots, g_\eta(x,y)\prod_{i\neq \eta} f_i(x)\rangle.
    \end{equation*}
\end{theorem}
\begin{proof}
    By the Chinese remainder theorem, we know that $2^{\text{rank}_{\mathbb{F}_2}(H)} = |\text{im}(H)|=|\psi(\text{im}(H))|=\prod_{j=1}^{\eta}|H_j|=2^{\ell m - \sum_{j=1}^\eta \textnormal{deg}_y g_j(x,y)\textnormal{deg}_x f_j(x)}$. By virtue of $k^\text{BB} = 2\ell m - 2\,\text{rank}_{\mathbb{F}_2}(H)$, we find that $k^\text{BB} = $$2\sum_{j=1}^\eta \textnormal{deg}_y g_j(x,y)\textnormal{deg}_x f_j(x)$.
\end{proof}

Here, we review a specific example where we calculate the code dimension $k^\text{BB}$ from one of the codes in Table 3. of Ref. \cite{bravyi}:

\begin{example}
Let us consider the $[[108,8,10]]$-code with the octet of constructors $\mathfrak{o} = \{6,9,3,1,2,3,1,2\}$. The CRT tells us that
\begin{equation*}
    \frac{\mathbb{F}_2[x]}{\langle x^9-1\rangle}\cong\frac{\mathbb{F}_2[x]}{\langle 1+x\rangle}\oplus \frac{\mathbb{F}_2[x]}{\langle 1+x+x^2\rangle}\oplus \frac{\mathbb{F}_2[x]}{\langle 1+x^3+x^6\rangle}.
\end{equation*}
This gives us an isomorphism $\psi:\mathcal{S}\to\bigoplus_{j=1}^\eta\mathcal{S}_j$ to smaller rings $\mathcal{S}_j$ in which we find the following polynomial decompositions
\begin{equation*}
\begin{aligned}
    y^6 - 1 &= (1+y)^2(1 + y + y^2)^2            &\quad &\text{in } \frac{\mathbb{F}_2[x]}{\langle 1 + x \rangle}[y],\\
    y^6 - 1 &= (1+y)^2(x + y)^2(1 + x + y)^2     &\quad &\text{in } \frac{\mathbb{F}_2[x]}{\langle 1 + x + x^2 \rangle}[y],\\
    y^6 - 1 &= (1+y)^2(x^3 + y)^2(1 + x^3 + y)^2 &\quad &\text{in } \frac{\mathbb{F}_2[x]}{\langle 1 + x^3 + x^6 \rangle}[y].
\end{aligned}
\end{equation*}
This yields the generators
\begin{equation*}
     \begin{aligned}
        g_1(x,y)&=\text{gcd}(\psi_1(a(x,y)),\psi_1(b(x,y)),y^6-1)=1, \\[1pt]
       g_2(x,y)&=\text{gcd}(\psi_2(a(x,y)),\psi_2(b(x,y)),y^6-1)=1+y+y^2, \\[1pt]
        g_3(x,y)&=\text{gcd}(\psi_3(a(x,y)),\psi_3(b(x,y)),y^6-1)=1,
    \end{aligned}
\end{equation*}
so that the image of the parity check matrices is generated by 
\begin{equation*}
    H=\langle (1+x+x^2)(1+x^3+x^6),(1+y+y^2)(1+x)(1+x^3+x^6),(1+x)(1+x+x^2)\rangle,
\end{equation*}
with code dimension $k^\text{BB}=2\,\text{deg}_y(1+y+y^2)\text{deg}_x(1+x+x^2)=8.$
\end{example}

Additionally, we can relax the condition that $\text{gcd}(\ell, 2) = 1$, so that $(\ell, m)$ can be \textit{any} pair of integers, also known as \textit{two-sided repeated-root} codes. In general, (quasi-)cyclic codes are must most cumbersome to analyse when their length is divided by the field characteristic over which they are defined.  The most direct way to calculate the code dimension $k$ is to construct a Gröbner basis and calculate its dimension as a generating set of ideals. First we define a total order for purposes of bivariate polynomial division.

\begin{definition}[Lexicographical ordering]
    Let $\mathcal{M}\subsetneq\mathcal{R}$ be a set of monomials. Then a \textit{monomial order} on the ring $\mathcal{R}$ is a total order $\prec$ such that
    \begin{itemize}
        \item $x\prec y \Rightarrow xz \prec yz$ for all $z\in\mathcal{M},$
        \item $1 \prec x$ for all $1\neq x\in\mathcal{M}.$
    \end{itemize}
    Then the \textit{lexicographical order} is the monomial order defined by
    \begin{equation*}
        x^{\alpha} \prec x^\beta \Leftrightarrow            \text{the first coordinates } \alpha_i \text{ and } \beta_i \text{ which are different satisfy } \alpha_i <\beta_i.
    \end{equation*}
    Hereby, we therefore fix  the order $1\prec y\prec y^2\prec \cdots\prec x\prec xy\prec \cdots \prec x^2 \prec \cdots$. The leading term/monomial of a polynomial $f(x,y)\in\mathcal{R}$ is denoted as $\text{LM}(f)$.
\end{definition} 

Imposing this monomial order, we can properly define polynomial division for bivariate polynomials.

\begin{definition}[Polynomial reduction]
    Let $g\in\mathcal{S}$ be a non-zero polynomial, and let $f,h\in\mathcal{S}$ be arbitrary polynomials. Then we say that $f$ reduces to $h$ modulo $g$ if $\text{LM}(g)$ divides any monomial $X$ in $f$, and $h=f-\frac{X}{\text{LM}(g)}g$. We write this as
    \begin{equation*}
        f\xrightarrow[]{g}h.
    \end{equation*}
    Given a set $G=\{g_1,\ldots, g_s\}\subseteq \mathcal{S}$, we say 
    \begin{equation*}
        f\xrightarrow[]{G}h
    \end{equation*}
    if there exists a subset $\{g_{i_1},\ldots,g_{i_t}\}\subseteq G$ such that
    \begin{equation*}
        f\xrightarrow[]{g_{i_1}}\cdots \xrightarrow[]{g_{i_t}} h.
    \end{equation*}
\end{definition}

\begin{definition}[Gröbner basis]
    A set of non-zero polynomials $G = \{g_1,\ldots, g_s\}$ is called a \textit{Gröbner basis} if for all non-zero $f\in I$, we have that $\text{LM}(g_i)$ divides $\text{LM}(f)$ for some $g_i\in G$. \textit{Buchberger's algorithm} provides a protocol that, given a lexicographical ordering, produces a Gröbner basis with respect to that ordering. The common zeroes of $g_1,\ldots,g_s$ are equal to the common zeroes of $f$ for all $f\in I.$
\end{definition}

\begin{example}
    Let $p(x)$ be a univariate polynomial. Then a Gröbner basis $G$ for $\{p(x),x^\ell-1\}$, for some $\ell\in\mathbb{N}$, is:
    \begin{equation*}
        G = \begin{cases}
            \{g(x)\} & \text{if } g(x):=\text{gcd}(p(x),x^\ell-1)\neq 1\\
            \{1\} & \text{otherwise}
        \end{cases}
    \end{equation*}
\end{example}
\begin{proof}
    If $\text{gcd}(p(x),x^\ell-1)\neq 1$, then $\langle x^\ell-1\rangle \subseteq \langle g(x)\rangle$, and $g(x)$ is the largest degree polynomial under lexicographic ordering with this property. Otherwise, due to Bézout's identity, we see that $p(x)$ and $x^\ell-1$ generate every single monomial, so a suitable Gröbner basis is simply $\{1\}.$
\end{proof}

Thus far, we have enough tools under our belt to introduce an explicit equation for $k^\text{BB}$. Using Gröbner bases, we can readily find the code dimension of two-sided repeated root BB codes:

\begin{theorem}\label{theorem:twosidedk}
    Let $G=\{g_1,\ldots,g_s\}$ be a Gröbner basis for $\{a(x,y),b(x,y),x^\ell-1,y^m-1\}$ and let $\textnormal{LM}(G):=\{\textnormal{LM}(g_1),\ldots,\textnormal{LM}(g_s)\}$ be the set of all leading monomials of the Gröbner basis. Let $\mathfrak{m}$ be the set of monomials not in the span of $\text{LM}(G)$. Then the code dimension of a BB code with generator polynomials $a(x,y),b(x,y)$ is given by
    \begin{equation*}
        k^\textnormal{BB} = 2\,\textnormal{dim}(\mathfrak{m}).
    \end{equation*}
\end{theorem}
\begin{proof}
    Since $k^\text{BB}=2\,\text{dim}(\text{ker}(A)\cap\text{ker}(B))$, we can simply consider the two-sided ideal $\text{ker}(A)\cap\text{ker}(B)$ on its own. In particular, elements of this ideal can be identified as elements of $\mathbb{F}_2[x,y]/\langle a(x,y),b(x,y),x^\ell-1,y^m-1\rangle.$ Thus a reduced Gröbner basis $G$ gives us a basis for the denominator, and by virtue of the Buchberger algorithm, any monomial which preceeds any of the leading monomials in $G$ under lexicographical ordering is not reducible under $G$, and its cardinality provides the size of the ideal $\text{ker}(A)\cap\text{ker}(B)$. The code dimension readily follows.
\end{proof}

\begin{example}
    Consider the code generated by $\mathfrak{o}=\{6,6,3,1,2,3,1,2\}$, then its dimension is 12 \cite{bravyi}. Using the Buchberger algorithm, we can construct a Gröbner basis \begin{equation*}
        G=\{y^3+x+x^2,1+y^3+x+xy+xy^2,1+y^2+y^4\}
    \end{equation*} 
    under lexicographical ordering. Its leading monomials span $\langle\text{LM}(G)\rangle=\langle x^2,xy^2,y^4\rangle$, so $\mathfrak{m}=\{1,x,y,y^2,y^3,xy\}$, and $2\,\text{dim}(\mathfrak{m})=12.$
\end{example}

We highlight that Theorem \ref{theorem:twosidedk} is closely related to one of the fundamental theorems of \textit{algebraic geometry}, namely the weak variant of Hilbert's Nullstellensatz:


\begin{theorem}[\cite{Hilbert}, Weak Nullstellensatz]
\label{nullstellensatz}    Hilbert's weak Nullstellensatz states that for any field $k$ and an algebraically closed field extension $K$ of $k$, an ideal $I\subseteq k[x_1,\ldots,x_n]$ contains 1 if and only if the polynomials in $I$ have no common zeros in $K^n$. 
\end{theorem}


Now we turn to \textit{code distance}, of which useful bounds exist, but it is difficult to obtain very tight bounds \cite{generalisedbicycle}. For coprime BB codes, we can exclude the following semi-trivial cases from our search for codes with a good distance by checking a simple set of equivalence relationships before calculations:

\begin{proposition}[Equal polynomials]\label{cor:d2}
Codes with equal trinomial constructors $A \equiv B$, are trivial (i.e. $d = 2$).
\end{proposition}
\begin{proof}
    If $A\equiv B$, then $[1|1]\in\text{ker}\, H_X$. The rowspace of $H_Z$ is generated by $\text{im}\, H_Z^\top = \lambda(z)[ a^\star(z)|a^\star(z)]$ for some polynomial $\lambda(z)\in\mathcal{R}$. Since $\text{gcd}(a^\star(z), z^{\ell m}-1)\neq 1$, $a^\star(z)$ has no inverse, so there exists no $\lambda(z)\in\mathcal{R}$ such that $[1|1]\in\text{im}\, H_Z^\top$. Thus, $[1|1]\in\text{ker}\, H_X\setminus\text{im}\, H_Z^\top$, and $d = 2$.
\end{proof}

\begin{proposition}[Squared polynomials]\label{cor:d4}
Codes with trinomial constructors that satisfy $A^2\equiv B$, have a distance that is always $d = 4$.
\end{proposition}
\begin{proof}
   If $A^2\equiv B$, then $[a(z)|1]\in\text{ker}\, H_X$, but since $g(z)$ in not invertible in $\mathcal{R}$, we have $[a(z)|1]\notin\text{im}\, H_Z^\top$. Since $|a(z)| =3$ by construct, the code distance satisfies $d = 4$.
\end{proof}

Recently, it has been shown that any abelian two-block group algebra (2BGA) code has a vanishing relative minimum distance \cite{generalisedbravyiterhal}. This proof was realised by showing that there exists a $D$-dimensional space in which all stabiliser qubits have their data qubits contained in a unit-radius ball with respect to a certain metric. We recall an important lemma:

\begin{lemma}[Local CSS code]\label{lemma:dim}
    Any generalised bicycle code with constant row weight $\Delta$ (i.e. independent of the code length $n$) is equivalent to a CSS code that is local in $D\leq \Delta-1$ dimensions, with code parameters satisfying the inequalities
    \begin{equation*}
        d\leq \mathcal{O}(n^{1-1/D})\quad\text{and}\quad kd^{2/(D-1)}\leq\mathcal{O}(n).
    \end{equation*}
\end{lemma}
\begin{proof}
    For a proof, we refer the reader to Ref. \cite{bravyidim, bravyidim2}.
\end{proof}

As a stronger requirement, the following theorem was posed as a generalisation of the Bravyi-Terhal bound that holds for any abelian 2BGA code of constant row weight $\Delta$:

\begin{theorem}\label{theorem:gen}[\cite{generalisedbravyiterhal}, Generalised Bravyi-Terhal bound]
    The minimum distance of any abelian 2BGA code of length $n$ and stabiliser weight $\Delta$, is bounded from above by
    \begin{equation*}
        d\leq 2\sqrt{\gamma_D}(\sqrt{D}+4)n^{1-1/D},
    \end{equation*}
    whenever $n^{1/D}\geq 8\sqrt{\gamma_D}$, where $D=\Delta-2$ and $\gamma_D$ is the Hermite constant.
\end{theorem}

\begin{theorem}[Asymptotic badness]\label{theorem:bad}
Any sequence of constant-weight BB codes of the form (\ref{eq:ansatz}) is asymptotically bad as we take $n\to\infty$.
\end{theorem}
\begin{proof}
    Since coprime BB codes have stabiliser weight 6, there exists an equivalent CSS code which is local in $D = \Delta -2 = 4$ dimensions. By virtue of Lemma \ref{lemma:dim}, we have
    \begin{equation*}
        \frac{k}{n}\leq \mathcal{O}(d^{-1/2})\quad\text{and}\quad\frac{d}{n}\leq\mathcal{O}(n^{-1/4}).
    \end{equation*}
    Thus, for $n\to\infty$, we have a vanishing relative minimum distance regardless of the constructors.
\end{proof}

This definitely excludes BB codes from being asymptotically good, though we stress that the utility of moderately long codes is not affected, as we will see later on. This leaves the\textit{ open question} whether non-abelian group algebras may generate codes within the generalised bicycle ansatz that produce asymptotically good parameters, or parameters at least as good as those of BB codes. 

We note that these bounds are irrelevant for BB codes of moderate length. Theorem \ref{theorem:gen} only applies for $n\geq 8192$, and the bound is trivially satisfied up until $n\approx 4\cdot 10^5$. Nevertheless, we can still adopt BB codes as an interesting demonstration for practical near-term applications of quantum error correction.

\section{Existence and constructions}\label{sec:construct}

All the previously mentioned theorems together comprise a full characterisation of BB code parameters. To efficiently explore the space of BB code parameters, we can trim down our search for trinomials $A$ and $B$ with good parameters by asking ourselves what polynomials can divide trinomials in the first place. For this reason, we turn to the factorisation of a cyclotomic polynomial $\Phi_p(x)$ over $\mathbb{F}_2$ and study its roots:

\begin{lemma}[\cite{trinomials}]
   For prime $p$ and $\Phi_p(x)=f_1(x)\cdots f_\eta(x)$ factorised into $\eta$ minimal polynomials, if any $f_i(x)$ divides a trinomial, then all of them do.
\end{lemma}
\begin{proof}
    Let $\alpha$ be a primitive root of $\Phi_p(x)$. All the roots of $f_1(x),\ldots,f_\eta(x)$ constitute the powers $\alpha,\alpha^2\ldots,\alpha^{p-1}$ and form a cyclic group under multiplication, i.e. $\alpha^p = 1$. All the $f_i(x)$'s share the same degree and primitivity $p$. Let $f_i(x)$ divide a trinomial $1+x^a+x^b$, for some $i\in[\eta]$. If we fix $\alpha$ to be a root of $f_i(x)$, then $1+\alpha^a+\alpha^b=0$. Let $\beta$ be a root of $f_j(x)$ for $j\neq i$, then we can reparameterise $\alpha$ as $\alpha=\beta^s$ for some $s$ such that gcd$(s,2) = 1$, i.e. $\alpha$ and $\beta$ are elements of distinct cyclotomic cosets. We then have $1+(\beta^s)^a+(\beta^s)^b=0$, so that $f_j(x)$ divides $1+x^{sa}+x^{sb}$. Since $i,j$ were arbitrary, all minimal polynomials either divide trinomials, or none of them do.
\end{proof}

\begin{lemma}\label{lemma:outlier}
    For prime $p$, $\Phi_p(x)=f_1(x)\cdots f_\eta(x)$ is factorised into $\eta$ minimal polynomials $f_i(x)$ that \textit{all} divide trinomials, in two cases only:
    \begin{itemize}
        \item $p$ is a Mersenne prime, i.e. of the form $p=2^s-1$ for some $s\in\mathbb{N}$,
        \item $p$ is an \textit{outlier}, of which the currently known members ($<3\cdot 10^6$) are $73, 121\, 369, 178\,481, 262\,657$ and $ 599\,479.$
    \end{itemize}
\end{lemma}
\begin{proof}
    If $p$ is a Mersenne prime, then $f_i(x)$ are primitive polynomials. Let $\alpha$ be a primitive root of $f_i(x)$, then all powers $1,\alpha,\alpha^2\ldots,\alpha^{p}$ are distinct, and constitute all elements of $\mathbb{F}^\times_{p+1}$. Hence there exists some pair of indices $(s,t)$ such that $1+\alpha^s+\alpha ^t=0$, and $f_i(x)$ is a divisor of $1+x^s+x^t$. Through an extensive computer search, the outliers were found, see Ref. \cite{trinomials} for more details.
\end{proof}

\begin{lemma}\label{lemma:div}
    For non-prime $n$, $\Phi_n(x)=f_1(x)\cdots f_\eta(x)$ is factorised into $\eta$ minimal polynomials $f_i(x)$ that all divide trinomials if and only if $p \mid n$ such that $\Phi_p(x) \mid t(x)$ for some trinomial $t(x)$.
\end{lemma}
\begin{proof}
    Let $n$ be an odd integer. Let $p,r$ be two primes such that $\text{gcd}(p,r)=1$, and let $m\in\mathbb{N}$ be an arbitrary positive integer. Then, two properties of cyclotomic polynomials are \begin{equation*}
    \Phi_{p^mr}(x)=\Phi_{pr}\left(x^{p^{m-1}}\right)
\quad\text{and}\quad
        \Phi_{pr}(x) = \frac{\Phi_p(x^r)}{\Phi_p(x)}.
    \end{equation*}
    These identities allow us to deduce the properties of a cyclotomic polynomial $\Phi_n(x)$ with composite index back to prime index polynomials $\Phi_p(x)$ for all $p\mid n$. Let $\overline{p}$ denote a Mersenne prime or an outlier prime. Let $f_j(x) \mid \Phi_{\overline{p}}(x)$ with associated cyclotomic coset $C_s$. Then there exists a $f'_j(x) \mid \Phi_n(x)$ that divides a trinomial. The roots of $\Phi_n(x)$ are described by the cyclotomic cosets $C_{rs}\setminus C_s$, which is non-empty by virtue of $r$ being coprime to the field characteristic 2. More specifically, if a trinomial $t(x)$ is divided by some minimal polynomial of $\Phi_{\overline{p}}(x)$, then $t(x^r)$ is divided by some minimal polynomial of $\Phi_{\overline{p}}(x^r)$. If $n$ were an even integer, the same argument holds by virtue of every minimal polynomial simply having a higher multiplicity.
\end{proof}

\begin{example}
    The integers 3 and 7 are Mersenne primes, and over $\mathbb{F}_2$ their cyclotomic polynomials read
    \begin{equation*}
        \Phi_3(x) = 1+x+x^2,\quad\text{and}\quad\Phi_7(x)=(1+x+x^3)(1+x^2+x^3),
    \end{equation*}
    which obviously have trinomial divisors. Examples of integers with Mersenne prime factors are 15 and 21, with cyclotomic polynomials
    \begin{equation*}
        \Phi_{15}(x)=(1+x+x^4)(1+x^3+x^4),\quad\Phi_{21}(x)=(1+x+x^2+x^4+x^6)(1+x^2+x^4+x^5+x^6).
    \end{equation*}
    The former obviously has trinomial divisors, but the latter is not as apparent, though this is guaranteed by Lemma \ref{lemma:div}. For instance, we find
    \begin{equation*}
        \text{gcd}(1+x^3+x^9,\Phi_{21}(x)) = 1+x^2+x^4+x^5+x^6.
    \end{equation*}
\end{example}

By understanding the structure of polynomials and their divisors, this leads us to the \textit{prime divisibility theorem}, which provides the unique conditions under which BB codes produce non-trivial codes:

\newpage

\begin{theorem}[Prime divisibility theorem]\label{theorem:primedivisibility}
    Any BB code under the trinomial ansatz, with polynomial constructors $a$ and $b$, of length $2\ell m$ is non-trivial $(k\geq 2)$ if and only if 
    \begin{enumerate}
        \item[] $\textnormal{(Coprime)}$
        \item $g(z):=\textnormal{gcd}(a(z),b(z),z^{\ell m}-1)\neq 1$,
        \item its code length $2\ell m$ is divided by a Mersenne prime or an outlier prime, defined in Lemma \ref{lemma:outlier}.
    \end{enumerate}
    or
    \begin{enumerate}
        \item[] \textnormal{(Quasi-cyclic)}
        \item $g:=\textnormal{gcd}(\ell,m)\neq 1$,
        \item there exist elements $X,Y$ in the algebraic closure of $\mathbb{F}_2$ such that $X^{\ell}=Y^{m}=1$ and $a(X,Y)=b(X,Y)=0$, or in other words, there exist roots $\zeta$ in the common splitting field of $x^\ell-1$ and $y^m-1$, and integers $i,j,k,I,J,K$ such that $\zeta^{gi}+\zeta^j+\zeta^k=0=\zeta^I+\zeta^{gJ}+\zeta^{gK}.$
        \item its code length $2\ell m$ is divided by a Mersenne prime or an outlier prime.
    \end{enumerate}
    
\end{theorem}
\begin{proof}
    For coprime BB codes, the result is trivial, since $\text{gcd}(a(z),b(z),z^{\ell m}-1)$ is only non-trivial if and only if $\ell m$ has a Mersenne prime or an outlier prime as a factor, by virtue of Lemma \ref{lemma:div}. For quasi-cyclic BB codes, we invoke the weak Nullstellensatz (Theorem \ref{nullstellensatz}) predicts a trivial code with Gröbner basis $G=\{1\}$ if there are no common zeroes among $a(x,y), b(x,y)$ and $x^\ell-1,y^m-1$. By virtue of $a,b$ being trinomials, there must exist a pair of roots $(X,Y)$ in the algebraic closure of $\mathbb{F}_2$ such that $a(X,Y)=b(X,Y)=0$, or in other words, they correspond to roots $\zeta$ in their common splitting field such that $\zeta$ satisfies a trinomial equation $1+\zeta^i+\zeta^j=0$ for some $(i,j)$, whose solutions can only exist if $\ell$ or $m$ is divided by a Mersenne prime or an outlier prime. The only requirement for common roots of $x^\ell-1$ and $y^m-1$ to exist is if $g:=\text{gcd}(\ell,m)\neq 1$. Let $\zeta$ be a root of $x^n-1$ as defined before, then there must exist integers $i,j,k,I,J,K$ such that
    \begin{equation*}
        \zeta^{ig}+\zeta^j+\zeta^k = 0 = \zeta^I+\zeta^{Jg}+\zeta^{Kg}.
    \end{equation*}
    This concludes the proof.
\end{proof}

Additionally, we can slightly constrain what code dimensions are allowed under the trinomial ansatz by noting that $k=2$ can not occur:

\begin{lemma}[$k\geq 4$ condition]
    Any non-trivial BB code satisfies $k\geq 4$.
\end{lemma}
\begin{proof}
    For coprime codes, we observe that the only minimal polynomial of degree 1 is $g(z)=1+z$ with root $z=1$. Trinomials can never have a root $z=1$, so that the code is trivial. For quasi-cyclic codes, the only Gröbner basis (under lexicographic ordering $x\succ y$) that would yield a complement $\mathfrak{m}$ of dimension $\text{dim}(\mathfrak{m})=1$ is $G=\{x+F(y),y+1\}$, where $F(y)$ is an arbitrary polynomial in $y$, which would imply $y=1$ being a root of both $a(x,y)$ and $b(x,y)$, which the trinomial ansatz disallows.
\end{proof}



An important aspect of considering the existence of BB codes, given a set of constructors $\mathfrak{o}$, is their connectivity, more specifically whether they yield a fully connected code. Ref. \cite{bravyi} has already pointed out that the novel codes they present \textit{are} fully connected, i.e. their Tanner graphs are not composed of several smaller separable code blocks with a lower code distance. To express this property more rigorously, they presented the following Lemma:

\begin{lemma}
    The Tanner graph of a coprime BB code is fully connected if and only if the set $S := \bigcup_{i,j\in\{1,2,3\}}\{A_iA_j^\top\}\cup\{B_iB_j^\top\}$ generates all possible monomials $\mathcal{M} = \{z^\mu\}_{\mu\in[\ell m]}$ or $\mathcal{M} = \{x^\mu y^\nu\}_{\mu\in[\ell],\nu\in[m]}$ if the code is coprime and quasi-cyclic respectively.
\end{lemma}
\begin{proof}
    A proof based on graph theory is presented in Ref. \cite{bravyi} under Lemma 3. Here, we simply presented the results in either coprime/quasi-cyclic monomial basis $\mathcal{M}.$ An intuitive explanation of this condition is that if a graph is irreducible, then any qubit must be connected to any other qubit through a continuous graph walk.
\end{proof}

For a more general discussion on disconnected 2BGA codes, we refer the reader to Ref. \cite{linpryadko}. Now, we show that this condition is equivalent to certain polynomials being \textit{incommensurate} with respect to the ideal $x^n-1$, and that this affects with generator polynomials are "allowed":

\begin{definition}[Incommensurate polynomial]
    Let $n$ be an arbitrary integer. Then a polynomial $g(z)$ is \textit{incommensurate} (with respect to $x^n-1$, over $\mathbb{F}_2$) if there exist no minimal polynomial $h_0(z)$ and odd constant $t\in\mathbb{N}$ such that $g(z)\equiv h_0(z^t)$ is coprime to $x^n-1$. 
\end{definition}

\begin{corollary}[Coprime connectivity]\label{theorem:connect}
    Let $n$ be an odd integer and let the constructor trinomials be $a(z) = 1+z^\alpha +z^\beta$ and $b(z)=1+z^\gamma+z^\delta$. Then the corresponding Tanner graph is fully connected if and only if $g(z)$ is incommensurate, or $\text{gcd}(\alpha,\beta,\gamma,\delta,n)=1$.
\end{corollary}
\begin{proof}
    Suppose $\text{gcd}(\alpha,\beta,\gamma,\delta,n)= g \neq 1$. Then, $z^g$ is a divisor of $A_i A_j^\top, B_i B_j^\top$ for all $i,j\in\{1,2,3\}.$ Thus $S$ generates $g\mathcal{M}:=\{z^{\mu g}\}_{\mu\in[n/g]} \subset \mathcal{M}$, and the Tanner graph is not fully connected.
\end{proof}

\begin{corollary}[Quasi-cyclic connectivity]\label{theorem:connect2}
    Let $n$ be an odd integer and let the constructor trinomials be $a(x,y) = x^{\alpha}+y^\beta +y^\gamma$ and $b(x,y)=y^{\delta}+x^\epsilon+x^\zeta$. Then the corresponding Tanner graph is fully connected if $g_1:=\text{gcd}(\alpha,\epsilon,\zeta,m)=1$ and $g_2:=\text{gcd}(\beta,\gamma,\delta,\ell)=1$.
\end{corollary}
\begin{proof}
    Suppose 'without loss of generality that $g_1>1$. Then $S=\{x^{g_1}\}$ generates a proper subgroup of $\mathcal{M}$, and the Tanner graph is disconnected. All monomials are therefore generated if and only if $g_1=g_2=1.$
\end{proof}

For coprime codes, we have found many different codes with different generator polynomials $g(z)$. However, upon searching for codes with a self-reciprocal $g(z)$, we only found that $g(z)=1+z+z^2$ suffices. The following Lemmas prove why this is the case.

\begin{lemma}\label{lemma:sym}
    Coprime BB codes with a self-reciprocal generator $g(z)$ satisfy $g(z) = \Phi_{3a}(z)$ for integer $a\in\mathbb{Z}$.
\end{lemma}
\begin{proof}
    Let $g(z)$ be a self-reciprocal polynomial that divides some trinomial $1+z^s+z^t$. By virtue of being self-reciprocal, any root $\alpha$ of $g(z)$ has a conjugate root $\alpha^{-1}$, so that
    \begin{equation*}
        1+\alpha ^s+\alpha^t = 0 =1+\alpha ^{-s}+\alpha^{-t},
    \end{equation*}
    which reduces to $\alpha^{t-s}=\alpha ^s$, or in other words, the order $e$ of $g(z)$ satisfies $e\mid |t-2s|.$ We can then write $t=2s+\lambda e$ for some $\lambda\in\mathbb{N}$, and parameterise the trinomial as
    \begin{equation*}
        1+z^s+z^{2s+\lambda e},
    \end{equation*}
    which can easily be checked to be a polynomial of order $e = 3s$. The only \textit{irreducible} trinomials of this form satisfy $s=3^a$ for $a\in\mathbb{N}$.
\end{proof}

\begin{lemma}[Coprime - self-reciprocal generator]
    Coprime BB codes with a symmetric generator $g(z)$ are connected if and only if $g(z)=1+z+z^2$.
\end{lemma}
\begin{proof}
    Let $g(z)$ be a self-reciprocal polynomial of order $e$, where $3 \mid e$. The only trinomials it divides are of the form $1+z^a+z^{2a+\lambda e}$ for some $\lambda\in\mathbb{N}$ by Theorem \ref{lemma:sym}. The gcd $g$ in Theorem \ref{theorem:connect} is equal to $a$. Thus, its Tanner graph is disconnected for $a\geq 2$, and connected only if $g(z)=1+z+z^2$.
\end{proof}

Consequently, we can only expect coprime BB codes generated by self-reciprocal polynomials if the generator is $g(z) = 1+z+z^2$. Comparing this to the list of known coprime BB codes \cite{coprime}, it explains why no other symmetric generator has ever been found.

\section{Code performance}\label{sec:performance}

Using our formalism, we have been able to construct new codes that were previously unknown. In particular, the divisibility condition (Theorem \ref{lemma:div}) has allowed us to explore new families. So far, every BB code has had either 3 or 7 as a divisor of $n$, which are very "digestible" parameters for experimentally implementing low-length to medium-length codes. It is therefore not surprising that no codes have emerged so far with divisors 31, 73, 127, 8191 etc. through- sheer brute force. The hardware constraints also impose limits on the optimality of codes that are expressible under the BB ansatz (\ref{eq:ansatz}), severely limiting the space of codes we can explore.

Now, we benchmark some new codes by turning to coprime BB codes, because of their simplicity. In general, we can easily pick two polynomials that share a certain divisor, using the divisibility criterion. For coprime BB codes, a general recipe goes as follows:
\begin{enumerate}
    \item Write out the complete irreducible factorisation of $z^{\ell m} - 1$ over $\mathbb{F}_2$.
    \item Pick any minimal polynomial $f_i(z)$ that divides a trinomial and makes the Tanner graph connected, and consider the constituent field $\mathbb{F}_2[z]/\langle f_i(z)\rangle.$ For simplicity, we can set $a(z)=f_i(z)$ \footnote{If $a(z)=f_i(z)r(z)$ for some $r(z) \nmid z^n-1$, then this statement is without loss of generality.}.
    \item Because of the cyclic nature of roots of unity in this field, we find $z^e = 1$, where $e$ is the order of $f_i(z)$. Multiplying any term in the base polynomial by this power ensures that $f_i(z)$ is still a divisor of this new polynomial. \textit{Example:}
    \begin{equation*}
        1+z+z^8=(1+z+z^2)(1+z^2+z^3+z^5+z^6).
    \end{equation*}
    This augmented polynomial is a candidate for $b(z)$.
    \item We can check a priori for certain low distance codes, such as the case $a(z)\equiv b(z), a(z)^2\equiv b(z)$ or $ b(z)/a(z)$ having a low weight. If a code is equivalent to a previously discarded code, we can discard it as well. A Monte Carlo simulation of the logical VS physical error rates yields a numerical approximation of the code distance. Early on during the simulation, we infer an estimate on the code distance, and reject low-distance codes.
    \item We keep generating new candidates, not equivalent to previous candidates, through multiplication by $z^e=1$ in either $a(z)$ or $b(z)$, until we have exhausted all possible candidates.
\end{enumerate}
This protocol also generalises for the quasi-cyclic case. The following tables, Table \ref{tab:1} for coprime codes and Table \ref{tab:2} for quasi-cyclic codes, present some of the best codes we have been able to find using our method. We also explicitly show codes with divisors that have never been shown before, such as 31 or 73. In Fig. \ref{fig:hist}, a collection of known coprime codes are shown, reflecting that for large code lengths $n$, we are less likely to find codes with high error-correcting capabilities. The next section explains how the code distance was retrieved from numerical simulations.

\begin{figure*}
    \centering{
    \makebox[\textwidth][c]{%
        \includegraphics[width=1.2\textwidth]{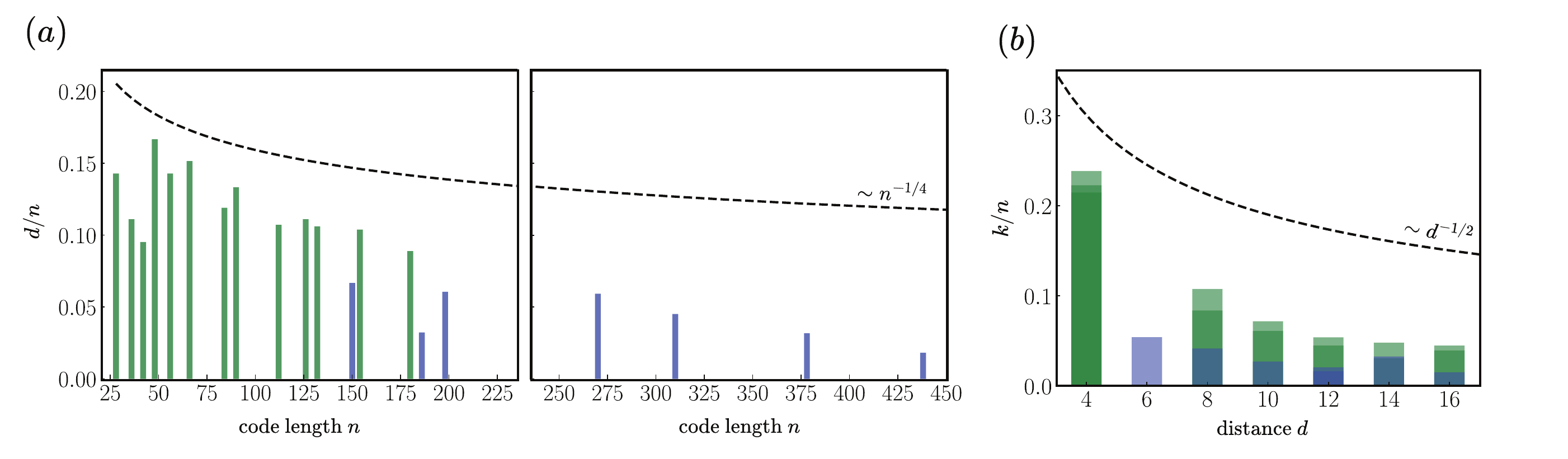}
    }
\captionsetup{justification=Justified}
 \caption{\textbf{(a)} Histograms of coprime BB codes, where the relative minimum distance $d/n$ is graphed against the code length $n$. Known codes (see Ref. \cite{coprime}) are coloured green, and our new codes are highlighted in blue. The asymptotic scaling (dashed line) is plotted to show the behaviour upper bound, and serves as a highlight of the asymptotic behaviour rather than a strict inequality. For larger code lengths, it appears harder to find codes with a relatively high relative minimum distance. \textbf{(b)} Histogram of the rate $k/n$ graphed as a function of the distance $d$. Overlaid colours indicate that multiple codes of the same distance have different rates.}
    \label{fig:hist}
    }
\end{figure*}

\begin{figure}[t]
    \centering{
    \includegraphics[width=0.66\textwidth]{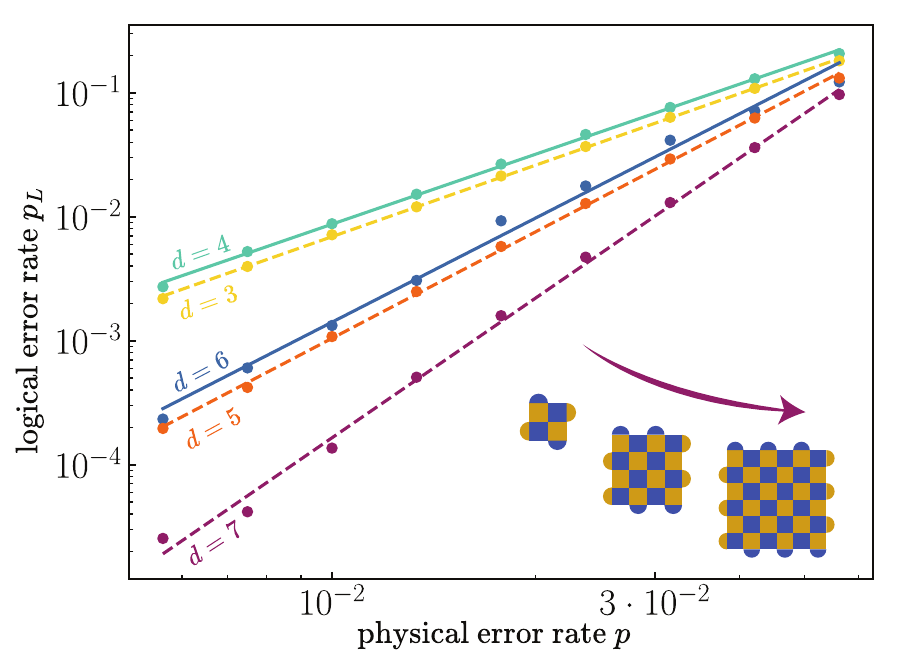}
    \captionsetup{justification=Justified}
    \caption{Comparison of equivalent codes between the $[[30,4,4]]$ and $[[30,4,6]]$ BB codes (highlighted in solid lines), and 4 patches of the $[[9,1,3]]$, $[[25,1,5]]$ and $[[49,1,7]]$ surface codes (highlighted in dashed lines). Codes with equivalent slopes will yield roughly the same performance per logical qubit, though BB codes trade physical qubit overhead for harder connectivity constraints. As an example, the $[[30,4,6]]_\text{BB}$ code and 4 patches of the $[[25,1,5]]_\text{SC}$ code have the same error correcting capabilities, while the latter requires 100 data qubits. The physical error rate range was chosen for convenience of the fit, such that we are below the error rate threshold, but not too low such that Monte Carlo errors skew the results too gravely. The surface code was decoded using MWPM, BB codes were decoded using BP.}
    \label{fig:logbothai}
    }
\end{figure}

\begin{table}[!h]
\begin{center}
\caption{Partial list of new coprime codes we have found. Because of finite size inaccuracies, code distance is accurate up to $\pm 2$.}
\begin{tabular}{|| c | c | c | c | c | c | c||} 
 \hline 
 $\ell$ & $m$ & qubits & $a(z)$ & $b(z)$ & $k$ & $d$\\ [0.5ex] 
 \hline\hline
 3 & 25 & 150 & $1+z+z^2$ & $1+z^2+z^{16}$ & 4 & 10\\
\hline
 9 & 11 & 198 & $1+z+z^2$ & $1+z^5+z^{37}$ & 4 & 12\\
 \hline
  5 & 27 & 270 & $1+z+z^2$ & $1+z^2+z^{25}$ & 4 & 16\\
 \hline
 7 & 27 & 378 & $1+z+z^3$ & $1+z+z^{31}$ & 6 & 12\\
 \hline
 3 & 31 & 186 & $1+z^2+z^5$ & $1+z^2+z^{36}$ & 10 & 6\\
 \hline
 5 & 31 & 310 & $1+z^2+z^5$ & $1+z^5+z^{64}$ & 10 & 14\\
 \hline
 11 & 31 & 682 & $1+z^2+z^5$ & $1+z^2+z^{67}$ & 10 & 14\\
 \hline
 3 & 73 & 438 & $1+z+z^9$ & $1+z^9+z^{74}$ & 18 & 8\\
 \hline
 5 & 73 & 730 & $1+z+z^9$ & $1+z+z^{82}$ & 18 & 10\\
 \hline
 
\end{tabular}
\label{tab:1}
\end{center}

\end{table}


\begin{table}[!h]
\begin{center}
\caption{Partial list of new quasi-cyclic codes we have found. Because of finite size inaccuracies, code distance is accurate up to $\pm 2$. In \textbf{bold}, we highlighted the smallest error-detecting and smallest error-correcting BB codes that exist.}
\begin{tabular}{|| c | c | c | c | c | c | c||} 
 \hline 
 $\ell$ & $m$ & qubits & $a(x,y)$ & $b(x,y)$ & $k$ & $d$\\ [0.5ex] 
 \hline\hline
 $\boldsymbol{3}$ & $\boldsymbol{3}$ & $\boldsymbol{18}$ & $\boldsymbol{1+y+y^2}$ & $\boldsymbol{y+1+x}$ & $\boldsymbol{4}$ & $\boldsymbol{2}$\\
\hline
3 & 3 & 18 & $1+y+y^2$ & $1+x+x^2$ & 8 & 2\\
\hline
$\boldsymbol{6}$ & $\boldsymbol{3}$ & $\boldsymbol{36}$ & $\boldsymbol{1+y+y^2}$ & $\boldsymbol{y+1+x}$ & $\boldsymbol{4}$ & $\boldsymbol{4}$\\
\hline
 7 & 6 & 84 & $1+y+y^2$ & $1+x+x^3$ & 12 & 4\\
\hline
 7 & 7 & 98 & $x^4+y+y^3$ & $y^4+x+x^3$ & 6 & 8\\
\hline
 14 & 7 & 196 & $x^4+y+y^3$ & $y^4+x+x^3$ & 6 & 12\\
\hline
15 & 15 & 450 & $x^5+y^3+y^4$ & $y^3+x^3+x^4$ & 8 & 18\\
\hline
 
\end{tabular}
\label{tab:2}
\end{center}

\end{table}

\subsection*{Performance of BB VS rotated surface code}

To benchmark the performance of these coprime BB codes, we compare their logical error rates to rotated surface codes of comparable code lengths $n$. We opt for 2 decoders: one employing the blossom algorithm using minimum weight perfect matching (MWPM) \cite{blossom}, and a decoder based on belief propagation (BP) \cite{belief}. The latter is most suitable for BB codes because of their non-trivial Tanner graph. From the low-error regime, the code distance can be inferred through the asymptotic relation
\begin{equation*}
    p_L \propto \left(\frac{p}{p_\text{th}}\right)^{t+1},
\end{equation*}
where $p$ denotes the physical error rate, $p_\text{th}$ is the error rate threshold, and $t$ is the number of errors we can correct. For BB codes, distances seem to be always even, so that $t + 1 = \frac{d}{2}$. We compare \textit{equivalent codes}, i.e. codes that share the same dimension $k$ and distance $d$, in terms of the number of physical qubits $n$ and the logical error rate $p_L$. Fixing $n^\text{BB}$, $k^\text{BB}$ and $d^\text{BB}$, the equivalent surface code (SC) would have a distance $d^\text{SC} = d^\text{BB}-1$ (since both give rise to the same number of correctable errors), and we need $k^\text{BB}$ surface code patches to match the number of logical qubits. This requires a total number of physical qubits 
\begin{equation*}\label{eq:compare}
n^\text{SC}=k^\text{BB}(d^\text{BB}-1)^2.
\end{equation*}
Fig. \ref{fig:logbothai} shows a very clear tradeoff between qubit number and connectivity: codes with a more complex graph connectivity will require less qubits to achieve the same number of logicals and error correction capabilities compared to planar rotated surface codes. Because of the square in Eq. (\ref{eq:compare}), BB codes are more economical when their code distance is high. For example, the $[[270,4,16]]_\text{BB}$ code requires about 3.3 times less qubits than 4 patches of the $[[225,1,15]]_\text{SC}$ code, while the $[[60,16,4]]_\text{BB}$ code requires about 2.4 times less qubits than 16 patches of the $[[9,1,3]]_\text{SC}$ code.

We reiterate that the disparity in performance between BB codes and planar topological error-correcting codes such as the surface code is a direct consequence of the higher connectivity, and therefore poses a trade-off:\textit{ the benefit in qubit overhead comes at the cost of non-local connectivity requirements}. There exist different solutions to the relatively difficult issue of implementing non-local connectivity in physical hardware. Either the non-locality is directly implemented in hardware, as has for example been demonstrated in qLDPC codes for trapped ions, using the Mølmer-Sørensen gate or variants \cite{trappedionlong, molmersorensen}, neutral atoms, using Rydberg states that interact over long distances \cite{fold, pecorari}, and fusion-based quantum computers \cite{fusion}, to name a few examples. Alternatively, the non-locality is mimicked through a qubit remapping with easier implementation \cite{chenplanar, fold}, or for example by making use of the thickness-2 property of BB codes to split them into two planar codes \cite{bravyi, jj}.

\section{Summary}\label{sec:summary}
Bivariate bicycle (BB) codes are a recent proposal for low-overhead quantum memory \cite{bravyi}, but so far, the trade-off between $[[n,k,d]]$ has been studied little, except for the case where $\ell$ and $m$, the system size, are coprime \cite{coprime}. Brute force calculations for finding existing non-trivial codes, and their code distance, are very inefficient. Additionally, some more efficient algorithms only provide upper bounds, meaning they can only exclude trivial codes early on from their search, while still brute forcing the rest.\\

In this Article, we have developed a novel and more efficient search algorithm to study BB codes that narrows down the search for constructors $\mathfrak{o}$ that guarantee non-trivial codes, \textit{regardless} of $\ell,m$. We used it to find a series of new codes, some of which have a higher number of logical qubits $k$ than previously known constructions. To answer the main question of this paper, we have provided several theorems. We have shown what number of physical qubits $n$ is required for codes to exist in Theorem \ref{theorem:primedivisibility}, under what circumstances BB codes are fully connected in Corollary \ref{theorem:connect} and Corollary \ref{theorem:connect2}, we have highlighted distance bounds and consequently demonstrated asymptotic badness in Theorem \ref{theorem:bad}. We have therefore shown a significant difference between BB codes and the codes predicted in the landmark results of Panteleev and Kalachev in Ref. \cite{ldpc}, due to the abelian nature of BB codes.\\

Though BB codes are excluded from the search of good practical codes due to asymptotic badness, our results do not affect the utility of moderately long codes. In fact, Lemma \ref{lemma:dim} predicts that a good cyclic code of this ansatz would require arbitrarily large connectivity as $n\to\infty$, yielding impractical implementation. Present-day quantum processors employing $\lessapprox 1000$ qubits could already implement most of the BB codes that are known so far, for example using the protocol outlined in Ref. \cite{chenplanar}. Our methods are also beneficial for designing cyclic LDPC codes for present-day quantum processor sizes and architectures, and they could be used to develop sequences of codes that are the best with respect to either $k/n$ or $d/n$ to date, paving the way for experiments demonstrating error correction routines outperforming the surface code.\\

The next line of research could be to investigate whether non-abelian 2BGA codes allow for asymptotically good sequences of codes, or at least allow for superior code parameters than BB codes of equal length $n$, while maintaining a bounded parity check row weight. This question has been partially answered in Ref. \cite{linpryadko}, though many open problems still persist, such as under what conditions non-abelian codes outperform their abelian counterparts for equal code lengths and group size.

\section*{Acknowledgments}
We thank Fabrizio Conca, Pascal Otjens, Jyrki Lahtonen and Raul Parcelas Resina dos Santos for fruitful discussions. This research is financially supported by the Dutch Ministry of Economic Affairs and Climate Policy (EZK), as part of the Quantum Delta NL programme, and by the Netherlands Organisation for Scientific Research (NWO) under Grant No. 680.92.18.05, and by the the Horizon Europe program HORIZON-CL4-2021-DIGITAL-EMERGING-01-30 via the project 101070144 (EuRyQa).

\section*{Data Availability}
The data that support the findings of this study are available from the corresponding author upon request.

\bibliographystyle{unsrt}
\bibliography{Bibliography}

\end{document}